\documentclass[a4paper]{llncs}
\usepackage{subfig}
\usepackage{graphicx}
\usepackage{url}
\usepackage{tikz}

\usepackage{times}
\usepackage{amssymb,amsmath,amsfonts}
\usepackage{latexsym}
\usepackage{verbatim}

\usepackage{enumerate}
\usepackage{graphicx}
\usepackage{amssymb}
\usepackage{tikz}

\DeclareMathOperator{\Inc}{Inc}
\DeclareMathOperator{\LG}{LG}

\newcommand{\F}{\ensuremath{\mathcal{F}}}

\renewcommand{\S}{\ensuremath{\mathcal{S}}}

\newcommand{\Prof}{\ensuremath{\mathcal{P}}}

\newcommand{\Leaves}{\ensuremath{\mathcal{L}}}

\newtheorem{observation}{Observation}

\begin{document}
\mainmatter  
\title{Characterizing Compatibility and Agreement of Unrooted Trees via Cuts in Graphs\thanks{This work was supported in part by the National Science Foundation under grants CCF-1017189 and DEB-0829674.}}
\titlerunning{Compatibility and Agreement of Unrooted Trees}
\author{Sudheer Vakati \and David Fern\'{a}ndez-Baca}
\authorrunning{S. Vakati \and D. Fern\'{a}ndez-Baca}
\institute{Department of Computer Science, Iowa State University, Ames, IA\ \ 50011, USA \\\email{\{svakati,fernande\}@iastate.edu}}

\maketitle

\begin{abstract}
Deciding whether there is a single tree ---a supertree--- that summarizes the evolutionary information in a collection of unrooted trees is a fundamental problem in phylogenetics.  We consider two versions of this question: agreement and compatibility.  In the first, the supertree is required to reflect precisely the relationships among the species exhibited by the input trees.  In the second, the supertree can be more refined than the input trees. 

Tree compatibility can be characterized in terms of the existence of a specific kind of triangulation in a structure known as the display graph. Alternatively, it can be characterized as a chordal graph sandwich problem in a structure known as the edge label intersection graph. Here, we show that the latter characterization yields a natural characterization of compatibility in terms of minimal cuts in the display graph, which is closely related to compatibility of splits. We then derive a characterization for agreement.
\end{abstract}

\section{Introduction}
A \emph{phylogenetic tree} $T$ is an unrooted tree whose leaves are bijectively mapped to a label set $\Leaves(T)$. Labels represent species and $T$ represents the evolutionary history of these species.  Let $\Prof$ be a collection of phylogenetic trees. We call $\Prof$ a \emph{profile}, refer to the trees in $\Prof$ as \emph{input trees}, and denote  the combined label set of the input trees, $\bigcup_{T \in \Prof}\Leaves(T)$, by $\Leaves(\Prof)$. A \emph{supertree} of $\Prof$ is a phylogenetic tree whose label set is $\Leaves(\Prof)$. The goal of constructing a supertree for a profile is to synthesize the information in the input trees in a larger, more comprehensive, phylogeny \cite{Gordon86}.  Ideally, a supertree should faithfully reflect the relationships among the species implied by the input trees.  In reality, it is rarely possible to achieve this, because of conflicts among the input trees due to errors in constructing them or to biological processes such as lateral gene transfer and gene duplication. 

We consider two classic versions of the supertree problem, based on the closely related notions of compatibility and agreement.  Let $S$ and $T$ be two phylogenetic trees where $\Leaves(T) \subseteq \Leaves(S)$ ---for our purposes, $T$ would be an input tree and $S$ a supertree.  Let $S'$ be the tree obtained by suppressing any degree two vertices in the minimal subtree of $S$ connecting the labels in $\Leaves(T)$. We say that $S$ \emph{displays} $T$, or that $T$ and $S$ are \emph{compatible}, if $T$ can be derived from $S'$ by contracting edges. We say that tree $T$ is an \emph{induced subtree} of $S$, or that $T$ and $S$ \emph{agree}, if $S'$ is isomorphic to $T$. 

Let $\Prof$ be a profile. The \emph{tree compatibility problem} asks if there exists a supertree for $\Prof$ that displays all the trees in $\Prof$. If such a supertree $S$ exists, we say that $\Prof$ is \emph{compatible} and $S$ is a \emph{compatible supertree} for $\Prof$. 
The \emph{agreement supertree problem} asks if there exists a supertree for $\Prof$ that agrees with all the trees in $\Prof$.   If such a supertree $S$ exists, we say that $S$ is an \emph{agreement supertree} (AST) for $\Prof$. 

Compatibility and agreement embody different philosophies about conflict.  An agreement  supertree must reflect precisely the evolutionary relationships exhibited by the input trees.  In contrast, a compatible supertree is allowed to exhibit more fine-grained relationships among certain labels than those exhibited by an input tree.  Note that compatibility and agreement are equivalent when the input trees are binary.

If all the input trees share a common label (which can be viewed as a root node), both tree compatibility and agreement are solvable in polynomial time~\cite{AhoSagivSzymanskiUllman81,NgWormald96}.  In general, however, the two problems are NP-complete, and remain so even when the trees are quartets; i.e., binary trees with exactly four leaves~\cite{Steel92}. Nevertheless, Bryant and Lagergren showed that the tree compatibility problem is fixed parameter tractable when parametrized by number of trees~\cite{BryantLagergren06}.  
It in unknown whether or not the agreement supertree problem has the same property.  


To prove the fixed-parameter tractability of tree compatibility, Bryant and Lagergren first showed that a necessary (but not sufficient) condition for a profile to be compatible is that the tree-width of a certain graph ---the \emph{display graph} of the profile (see Section~\ref{sec:DG_ELIG})--- be bounded by the number of trees.  They then showed how to express compatibility as a bounded-size monadic second-order formula on the display graph.  By Courcelle's Theorem~\cite{Courcelle90,ArnborgLagergrenSeese91}, these two facts imply that compatibility can be decided in time linear in the size of the display graph.  Unfortunately,  Bryant and Lagergren's argument amounts essentially to only an existential proof, as it is not clear how to obtain an explicit algorithm for unrooted compatibility from it  

A necessary step towards finding a practical algorithm for compatibility ---and indeed for agreement--- is to develop an explicit characterization of the problem.  In earlier work~\cite{Vakati11}, we made some progress in this direction, characterizing tree compatibility in terms of the existence of a legal triangulation of the display graph of the profile. Gysel et al.~\cite{Gysel2012} provided an alternative characterization, based on a structure they call the edge label intersection graph (ELIG) (see Section~\ref{sec:DG_ELIG}).  Their formulation is in some ways simpler than that of~\cite{Vakati11}, allowing Gysel et al.~to express tree compatibility as a chordal sandwich problem.  Neither~\cite{Vakati11} nor~\cite{Gysel2012} deal with agreement.

Here, we show that the connection between separators in the ELIG and cuts in the display graph  (explored in Section~\ref{sec:DG_ELIG}) leads to a new, and natural, characterization of compatibility in terms of minimal cuts in the display graph (Section~\ref{sec:comp_cuts}). We then show how such cuts are closely related to the splits of the compatible supertree (Section~\ref{sec:splits_cuts}). Lastly, we give a characterization of the agreement in terms of minimal cuts of the display graph (Section~\ref{sec:agree_cuts}). 
To our knowledge, there was no previous characterization of the agreement supertree problem for unrooted trees.


\section{Preliminaries}

\paragraph{Splits, Compatibility, and Agreement}

A \emph{split} of a label set $L$ is a bipartition of $L$ consisting of non-empty sets. We denote a split $\{X, Y\}$ by $X|Y$. Let $T$ be a phylogenetic tree. Consider an internal edge $e$ of $T$. Deletion of $e$ disconnects $T$ into two subtrees $T_1$ and $T_2$. If $L_1$ and $L_2$ denote the set of all labels in $T_1$ and $T_2$, respectively, then $L_1 | L_2$ is a split of $\Leaves(T)$. We denote by $\sigma_e(T)$ the split corresponding to edge $e$ of $T$ and by $\Sigma(T)$ the set of all splits corresponding to all internal edges of $T$. 

We say that a tree $T$ \emph{displays} a split $X|Y$ if there exists an internal edge $e$ of $T$ where $\sigma_e(T) = X|Y$. A set of splits is \emph{compatible} if there exists a tree that displays all the splits in the set. It is well-known that two splits $A_1|A_2$ and $B_1|B_2$ are compatible if and only if at least one of $A_1 \cap B_1$, $A_1 \cap B_2$, $A_2 \cap B_1$ and $A_2 \cap B_2$ is empty~\cite{SempleSteel03}. 

\begin{theorem}[Splits-Equivalence Theorem~\cite{Buneman71,SempleSteel03}] \label{thm:SET} Let $\Sigma$ be a collection of non-trivial splits of a label set $X$. Then, $\Sigma = \Sigma(T)$ for some phylogenetic tree $T$ with label set $X$ if and only if the splits in $\Sigma$ are pairwise compatible.  Tree $T$ is unique up to isomorphism.
\end{theorem}

Let $S$ be a phylogenetic tree and let $Y$ be a subset of $\Leaves(S)$.  Then, $S_{|Y}$ denotes the tree obtained by suppressing any degree-two vertices in the minimal subtree of $S$ connecting the labels in $Y$.  Now, let $T$ be a phylogenetic tree such that $\Leaves(T) \subseteq \Leaves(S)$. Then, $S$ \emph{displays} $T$ if and only if $\Sigma(T) \subseteq \Sigma(S_{|\Leaves(T)})$; $T$ and $S$ \emph{agree} if and only if $\Sigma(T) = \Sigma(S_{|\Leaves(T)})$. 

\paragraph{Cliques, Separators, Cuts, and Triangulations.}
 Let $G$ be a graph. We represent the vertices and edges of $G$ by $V(G)$ and $E(G)$ respectively. A \emph{clique} of $G$ is a complete subgraph of $G$. A clique $H$ of $G$ is \emph{maximal} if there is no other clique $H'$ of $G$ where $V(H) \subset V(H')$. For any $U \subseteq V(G)$, $G-U$ is the graph derived by removing vertices of $U$ and their incident edges from $G$.  For any $F \subseteq E(G)$, $G-F$ is the graph with vertex set $V(G)$ and edge set $E(G) \setminus F$. 

For any two nonadjacent vertices $a$ and $b$ of $G$, an $a$-$b$ $separator$ of $G$ is a set $U$ of vertices where $U \subset V(G)$ and
$a$ and $b$ are in different connected components of $G-U$. An $a$-$b$ separator $U$ is \emph{minimal} if for every $U' \subset U$, $U'$ is not an $a$-$b$ separator. A set $U \subseteq V(G)$ is a \emph{minimal separator} if $U$ is a minimal $a$-$b$ separator for some nonadjacent vertices $a$ and $b$ of $G$.  We represent the set of all minimal separators of graph $G$ by $\triangle_G$.  
Two minimal separators $U$ and $U'$ are \emph{parallel} if $G-U$ contains at most one component $H$ where $V(H) \cap U' \neq \emptyset$.

A connected component $H$ of $G-U$ is \emph{full} if for every $u \in U$ there exists some vertex $v \in H$ where $\{u, v\} \in E(G)$.

\begin{lemma}[\cite{Parra1997}]
\label{lm:two_full}
For a graph $G$ and any $U \subset V(G)$, $U$ is a minimal separator of $G$ if and only if $G-U$ has  at least two full components.
\end{lemma}

A \emph{chord} is an edge between two nonadjacent vertices of a cycle. A graph $H$ is \emph{chordal} if and only if every cycle of length four or greater in $H$ has a chord. A chordal graph $H$ is a \emph{triangulation} of graph $G$ if $V(G) = V(H)$ and $E(G) \subseteq E(H)$.  The edges in $E(H) \setminus E(G)$ are called \emph{fill-in} edges of $G$.   A triangulation is \emph{minimal} if removing any fill-in edge yields a non-chordal graph.

A \emph{clique tree} of a chordal graph $H$ is a pair $(T, B)$ where (i) $T$ is a tree, (ii) $B$ is a bijective function from vertices of $T$ to maximal cliques of $H$, and (iii) for every vertex $v \in H$, the set of all vertices $x$ of $T$ where $v \in B(x)$ induces a subtree in $T$.  Property (iii) is called \emph{coherence}.

Let $\F$ be a collection of subsets of $V(G)$. We represent by $G_{\F}$ the graph derived from $G$ by making the set of vertices of $X$ a clique in $G$ for every $X \in \F$. The next result summarizes basic facts about separators and triangulations (see~\cite{Bouchitte2001,Heggernes2006,Parra1997}).

\begin{theorem}
\label{thm:parallel_minimal_ct}
Let $\F$ be a maximal set of pairwise parallel minimal separators of $G$ and $H$ be a minimal triangulation of $G$. Then, the following statements hold.
\begin{enumerate}[(i)]
\item $G_{\F}$ is a minimal triangulation of $G$.
\item Let $(T, B)$ be a clique tree of $G_{\F}$. There exists a minimal separator $F \in \F$ if and only if there exist two adjacent vertices $x$ and $y$ in $T$ where $B(x) \cap B(y) = F$.
\item $\triangle_{H}$ is a maximal set of pairwise parallel minimal separators of $G$ and $G_{\triangle_H} = H$.
\end{enumerate}
\end{theorem}

A \emph{cut} in a connected graph $G$ is a subset $F$ of edges of $G$ such that $G-F$ is disconnected. A cut $F$ is \emph{minimal} if there does not exist $F' \subset F$ where  $G-F'$ is disconnected. Note that if $F$ is minimal, $G-F$ has exactly be two connected components. Two minimal cuts $F$ and $F'$ are \emph{parallel} if $G-F$ has at most one connected component $H$ where $E(H) \cap F' \neq \emptyset$. 

\section{Display Graphs and Edge Label Intersection Graphs}

\label{sec:DG_ELIG}

We now introduce the two main notions that we use to characterize compatibility and agreement: the display graph and edge label intersection graph.  We then present some known results about these graphs, along with new results on the relationships between them.  Here and in the rest of the paper, $[m]$ denotes the set $\{1, \dots, m\}$, where $m$ is a non-negative integer.  Since for any phylogenetic tree $T$ there is a bijection between the leaves of $T$ and $\Leaves(T)$, we refer to the leaves of $T$ by their labels.  
 
Let $\Prof = \{T_1, T_2, \cdots, T_k\}$ be a profile. We assume that for any $i, j \in [k]$ such that $i \neq j$, the sets of internal vertices of input trees $T_i$ and $T_j$ are disjoint. The \emph{display graph} of $\Prof$, denoted by $G(\Prof)$, is a graph whose vertex set is $\bigcup_{i \in [k]} V(T_i)$ and edge set is $\bigcup_{j \in [k]} E(T_j)$ (see Fig.~\ref{fig:example}). A vertex $v$ of $G(\Prof)$ is a \emph{leaf} if $v \in \Leaves(\Prof)$. Every other vertex of $G(\Prof)$ is an \emph{internal}. An edge of $G(\Prof)$ is \emph{internal} if its endpoints are both internal. If $H$ is a subgraph of $G(\Prof)$, then $\Leaves(H)$ represents the set of all leaves of $H$. 

A triangulation $G'$ of $G(\Prof)$ is \emph{legal} if it satisfies the following conditions.
\begin{enumerate} [(LT1)]
\item For every clique $C$ of $G'$, if $C$ contains an internal edge, then it cannot contain any other edge of $G(\Prof)$.  
\item There is no fill-in edge in $G'$ with a leaf as an endpoint.
\end{enumerate}

\begin{theorem}[Vakati, Fern\'{a}ndez-Baca~\cite{Vakati11}]
\label{thm:lt}
A profile $\Prof$ of unrooted phylogenetic trees is compatible if and only if $G(\Prof)$ has a legal triangulation.
\end{theorem}

In what follows, we assume that $G(\Prof)$ is connected. If it is not, the connected components of $G(\Prof)$ induce a partition of $\Prof$ into sub-profiles such that for each sub-profile $\Prof'$, $G(\Prof')$ is a connected component of $G(\Prof)$. It is easy to see that $\Prof$ is compatible if and only if each sub-profile is compatible. 

The \emph{edge label intersection graph of $\Prof$}, denoted $\LG(\Prof)$, is the line graph of $G(\Prof)$~\cite{Gysel2012}.\footnote{Note that Gysel et al.~refer to $\LG(\Prof)$ as the modified edge label intersection graph~\cite{Gysel2012}.} That is, the vertex set of $\LG(G)$ is $E(G(\Prof))$ and two vertices of $\LG(\Prof)$ are adjacent if the corresponding edges in $G(\Prof)$ share an endpoint. For an unrooted tree $T$, $\LG(T)$ denotes $\LG(\{T\})$. 


\begin{observation}\label{obs:path}
Let $F$ be a set of edges of $G(\Prof)$ and let $\{v_1, v_2,\dots,v_m\} \subseteq V(G(\Prof))$ where $m \geq 2$. Then, $v_1, v_2, \dots, v_m$ is a path in $G(\Prof)-F$ if and only if $\{v_1, v_2\},$ $ \dots,$ $\{v_{m-1}, v_m\}$ is a path in in $\LG(\Prof)-F$.
\end{observation}

Thus, if $G(\Prof)$ is connected, so is $\LG(\Prof)$. Hence, in what follows, we assume that $\LG(\Prof)$ is connected.

A fill-in edge for $\LG(\Prof)$ is \emph{valid} if  for every $T \in \Prof$, at least one of the endpoints of the edge is not in $\LG(T)$. A triangulation $H$ of $\LG(\Prof)$ is \emph{restricted} if every fill-in edge of $H$ is valid.

\begin{theorem}[Gysel, Stevens, and Gusfield~\cite{Gysel2012}]
\label{thm:EL_compatibility}
A profile $\Prof$ of unrooted phylogenetic trees is compatible if and only if $\LG(\Prof)$ has a restricted triangulation.
\end{theorem}

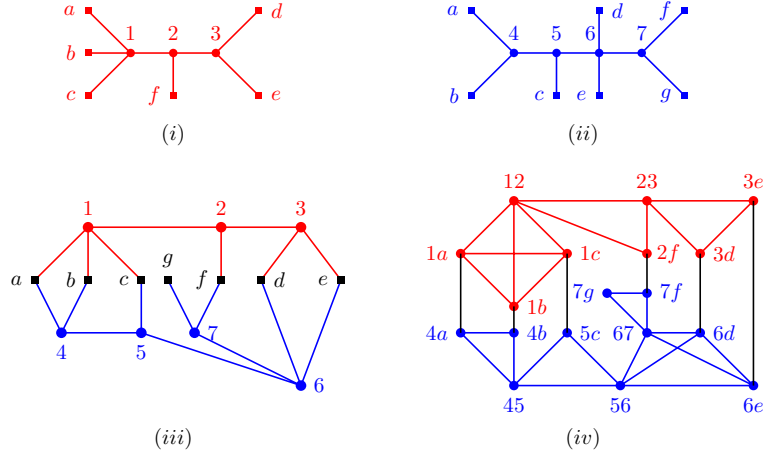
\begin{figure}[t!]
\renewcommand{\figurename}{Fig.}
 \begin{center}
\begin{tikzpicture} [inner sep = 2mm, semithick,scale=0.7]
\tikzstyle{every node}=[font=\fontsize{10}{10}, scale=0.8];

\coordinate (one_a) at (1,9);
\begin{scope}[red,scale=0.8]
\draw (one_a) --++(1, -1) coordinate(one_1) --++(-1, 0) coordinate(one_b) ++(1,0) --++(-1,-1) coordinate(one_c)  ++(1,1) --++(1,0) coordinate(one_2) --++(0,-1) coordinate(one_f) ++(0,1) --++(1,0) coordinate (one_3) --++(1,1) coordinate(one_d) ++(-1,-1) --+(1,-1) coordinate(one_e);
\path [fill] (one_1) circle(0.1) node[above] {$1$};
\path [fill] (one_2) circle(0.1) node[above] {$2$};
\path [fill] (one_3) circle(0.1) node[above] {$3$};
\path [fill] (one_a) node[left] {$a$} ++(-0.075,-0.075)  rectangle +(0.15,0.15);
\path [fill] (one_b) node[left] {$b$} ++(-0.075,-0.075)  rectangle +(0.15,0.15);
\path [fill] (one_c) node[left] {$c$} ++(-0.075,-0.075)  rectangle +(0.15,0.15);
\path [fill] (one_d) node[right]{$d$}++(-0.075,-0.075)  rectangle +(0.15,0.15);
\path [fill] (one_e) node[right]{$e$}++(-0.075,-0.075)  rectangle +(0.15,0.15);
\path [fill] (one_f) node[left]{$f$}++(-0.075,-0.075)  rectangle +(0.15,0.15);
\end{scope}

\path(one_f) ++(0, -0.75) coordinate(i)node{$(i)$};
\path (one_d) ++(4,0) coordinate(two_a);
\begin{scope}[blue,scale=0.8]
\draw (two_a) --++(1,-1) coordinate(two_4) --++(-1,-1) coordinate(two_b) ++(1,1) --++(1,0) coordinate(two_5)--++(0,-1) coordinate(two_c) ++(0,1) --++(1,0) coordinate(two_6) --++(0,1) coordinate(two_d) ++(0,-2)  coordinate(two_e)--++(0, 1) --++(1,0) coordinate(two_7) --++(1,1) coordinate(two_f) ++(-1,-1) --+(1,-1) coordinate(two_g);
\begin{scope}
\path [fill] (two_4) circle(0.1) node[above] {$4$};
\path [fill] (two_5) circle(0.1) node[above] {$5$};
\path [fill] (two_6) circle(0.1) +(-0.2,0) node[above] {$6$};
\path [fill] (two_7) circle(0.1) node[above] {$7$};
\end{scope}
\path [fill] (two_a) node[left] {$a$} ++(-0.075,-0.075)  rectangle +(0.15,0.15);
\path [fill] (two_b) node[left] {$b$} ++(-0.075,-0.075)  rectangle +(0.15,0.15);
\path [fill] (two_c) node[left] {$c$} ++(-0.075,-0.075)  rectangle +(0.15,0.15);
\path [fill] (two_d) node[right] {$d$} ++(-0.075,-0.075)  rectangle +(0.15,0.15);
\path [fill] (two_e) node[left] {$e$} ++(-0.075,-0.075)  rectangle +(0.15,0.15);
\path [fill] (two_f) node[left] {$f$} ++(-0.075,-0.075)  rectangle +(0.15,0.15);
\path [fill] (two_g) node[left] {$g$} ++(-0.075,-0.075)  rectangle +(0.15,0.15);
\end{scope}
\path (two_c) ++(0.5, -0.75) coordinate(ii)node{$(ii)$};
\path (one_c) ++(0,-2.5) coordinate(dis_1);

\begin{scope}[red]
\draw (dis_1) --++(-1, -1) coordinate(dis_a)  ++(1, 0) coordinate(dis_b) --++(0,1)--++(1,-1) coordinate(dis_c) ++(-1,1) --++(2.5,0) coordinate(dis_2) --++(0,-1) ++(0,1) --++(1.5,0) coordinate(dis_3)--++(-0.75,-1)coordinate(dis_d) ++(1.5,0) coordinate(dis_e) -- (dis_3);
\begin{scope}
\path [fill] (dis_1) circle(0.1) node[above] {$1$};
\path [fill] (dis_2) circle(0.1) node[above] {$2$};
\path [fill] (dis_3) circle(0.1) node[above] {$3$};
\end{scope}

\end{scope}

\begin{scope}[blue]
\draw (dis_a) --++(0.5,-1) coordinate(dis_4)--++(0.5,1) ++(-0.5,-1) --++(1.5,0) coordinate(dis_5)--++(0,1) ++(0,-1)--++(3,-1) coordinate(dis_6) --++(-2,1) coordinate(dis_7) --++(-0.5,1)coordinate(dis_g) ++(1,0) coordinate(dis_f) --++(-0.5,-1) (dis_6)--(dis_d) (dis_6)--(dis_e);

\begin{scope}
\path [fill] (dis_4) circle(0.1) node[below] {$4$};
\path [fill] (dis_5) circle(0.1) node[below] {$5$};
\path [fill] (dis_6) circle(0.1) node[right] {$6$};
\path [fill] (dis_7) circle(0.1) node[right] {$7$};
\end{scope}

\end{scope}

\begin{scope}[black]
\path [fill] (dis_a) node[left] {$a$} ++(-0.075,-0.075)  rectangle +(0.15,0.15);
\path [fill] (dis_b) node[left] {$b$} ++(-0.075,-0.075)  rectangle +(0.15,0.15);
\path [fill] (dis_c) node[left] {$c$} ++(-0.075,-0.075)  rectangle +(0.15,0.15);
\path [fill] (dis_d) node[right] {$d$} ++(-0.075,-0.075)  rectangle +(0.15,0.15);
\path [fill] (dis_e) node[left] {$e$} ++(-0.075,-0.075)  rectangle +(0.15,0.15);
\path [fill] (dis_f) node[left] {$f$} ++(-0.075,-0.075)  rectangle +(0.15,0.15);
\path [fill] (dis_g) node[above] {$g$} ++(-0.075,-0.075)  rectangle +(0.15,0.15);
\end{scope}

\path (i) ++(0, -5.75) node{$(iii)$};
\path (dis_3) ++(4,0.5) coordinate(el_12);

\begin{scope}[red]
\draw (el_12)--++(-1, -1) coordinate(el_1a) --++(1,-1) coordinate(el_1b) --++(1, 1) coordinate(el_1c) --(el_12) --(el_1b) (el_1a)--(el_1c);
\draw (el_12)--++(2.5,-1) coordinate(el_2f)--++(0,1) coordinate(el_23) -- (el_12);
\draw (el_23)--++(1,-1) coordinate(el_3d)--++(1,1) coordinate(el_3e)--(el_23);

\begin{scope}
\path [fill] (el_12) circle(0.09) node[above]{$12$};
\path [fill] (el_1a) circle(0.09) node[left]{$1a$};
\path [fill] (el_1c) circle(0.09) node[right]{$1c$};
\path [fill] (el_1b) circle(0.09) node[right]{$1b$};
\path [fill] (el_23) circle(0.09) node[above]{$23$};
\path [fill] (el_2f) circle(0.09) node[right, inner sep=1.5mm]{$2f$};
\path [fill] (el_3e) circle(0.09) node[above]{$3e$};
\path [fill] (el_3d) circle(0.09) node[right]{$3d$};
\end{scope}

\end{scope}

\path (el_1a) ++(0,-1.5) coordinate(el_4a);
\begin{scope}[blue]
\draw (el_4a) --++(1,0)coordinate(el_4b)--++(0,-1) coordinate(el_45) --(el_4a);
\draw (el_45)--++(1,1)coordinate(el_5c)--++(1,-1) coordinate(el_56) --(el_45);
\draw (el_56)--++(0.5,1)coordinate(el_67)--++(1,0) coordinate(el_6d)--++(1,-1)coordinate(el_6e) -- (el_56) -- (el_6d) (el_67)--(el_6e);
\draw (el_67) --++(0, 0.75)coordinate(el_7f)--++(-0.75,0) coordinate(el_7g)--(el_67);

\begin{scope}
\path [fill] (el_4a) circle(0.09) node[left]{$4a$};
\path [fill] (el_4b) circle(0.09) node[right]{$4b$};
\path [fill] (el_45) circle(0.09) node[below]{$45$};
\path [fill] (el_5c) circle(0.09) node[right]{$5c$};
\path [fill] (el_56) circle(0.09) node[below]{$56$};
\path [fill] (el_67) circle(0.09) node[left]{$67$};
\path [fill] (el_7g) circle(0.09) node[left]{$7g$};
\path [fill] (el_7f) circle(0.09) node[right]{$7f$};
\path [fill] (el_6d) circle(0.09) node[right]{$6d$};
\path [fill] (el_6e) circle(0.09) node[below]{$6e$};
\end{scope}
\end{scope}

\begin{scope}[black]
\draw (el_1a)--(el_4a);
\draw (el_1b)--(el_4b);
\draw (el_1c)--(el_5c);
\draw (el_3d)--(el_6d);
\draw (el_3e)--(el_6e);
\draw (el_2f)--(el_7f);
\end{scope}

\path (ii) ++(0,-5.75) node{$(iv)$};

\end{tikzpicture}
 \end{center}
\vspace{-0.3cm}
\caption{\emph{(i)} First input tree. \emph{(ii)} A second input tree, compatible with the first. \emph{(iii)} Display graph of the input trees. \emph{(iv)} Edge label intersection graph of the input trees; for every vertex, $uv$ represents edge $\{u,v\}$.}
 \label{fig:example}
\end{figure}

A minimal separator $F$ of $\LG(\Prof)$ is \emph{legal} if for every $T \in \Prof$, all the edges of $T$ in $F$ share a common endpoint; i.e., $F \cap E(T)$ is a clique in $\LG(T)$. 
The following theorem was mentioned in~\cite{Gysel2012}.  

\begin{theorem}\label{thm:edge_label}
A profile $\Prof$ is compatible if and only if there exists a maximal set $\F$ of pairwise parallel minimal separators in $\LG(\Prof)$ where every separator in $\F$ is legal.
\end{theorem}
\begin{proof}
Our approach is similar to the one used by Gusfield in~\cite{Gusfield10}.
Assume that $\Prof$ is compatible. From Theorem~\ref{thm:EL_compatibility}, there exists a restricted triangulation $H$ of $\LG(\Prof)$. We can assume that $H$ is minimal (if it is not, simply  delete fill-in edges repeatedly from $H$ until it is minimal). Let $\F = \triangle_{H}$. From Theorem~\ref{thm:parallel_minimal_ct}, $\F$ is a maximal set of pairwise parallel minimal separators of $\LG(\Prof)$ and $\LG(\Prof)_{\F}=H$. Suppose $\F$ contains a separator $F$ that is not legal. Let $\{e, e'\} \subseteq F$ where $\{e, e'\} \subseteq E(T)$ for some input tree $T$ and $e \cap e' = \emptyset$. The vertices of $F$ form a clique in $H$. Thus, $H$ contains  the edge $\{e, e'\}$. Since $\{e, e'\}$ is not a valid edge, $H$ is not a restricted triangulation, a contradiction. Hence, every separator in $\F$ is legal.

Let $\F$ be a maximal set of pairwise parallel minimal separators of $\LG(\Prof)$ where every separator in $\F$ is legal. From Theorem~\ref{thm:parallel_minimal_ct}, $\LG(\Prof)_{\F}$ is a minimal triangulation of $\LG(\F)$. If $\{e, e'\} \in E(\LG(\Prof)_{\F})$ is a fill-in edge, then $e \cap e' = \emptyset$ and there exists a minimal separator $F \in \F$ where $\{e, e'\} \subseteq F$. Since $F$ is legal, if $\{e, e'\} \subseteq E(T)$ for some input tree $T$ then $e \cap e' \neq \emptyset$. Thus,  $e$ and $e'$ are not both from $\LG(T)$ for any input tree $T$. Hence, every fill-in edge in $\LG(\Prof)_{\F}$ is valid, and $\LG(\Prof)_{\F}$ is a restricted triangulation. \qed
\end{proof}


Let $u$ of be a vertex of some input tree, Then, $\Inc(u)$ is the set of all edges of $G(\Prof)$ incident on $u$.  Equivalently, $\Inc(u)$ is the set of all vertices $e$ of $\LG(\Prof)$ such that $u \in e$.

Let $F$ be a cut of the display graph $G(\Prof)$.  $F$ is \emph{legal} if for every tree $T \in \Prof$, the edges of $T$ in $F$ are incident on a common vertex; i.e., if $F \cap E(T) \subseteq \Inc(u)$ for some $u \in V(T)$. $F$ is \emph{nice} if $F$ is legal and each connected component of $G(\Prof)-F$ has at least one edge.

\begin{lemma}\label{lm:cuts_seps}
Let $F$ be a subset of $E(G(\Prof))$. Then, $F$ is a legal minimal separator of $\LG(\Prof)$ if and only if $F$ is a nice minimal cut of $G(\Prof)$.
\end{lemma}

To prove the Lemma~\ref{lm:cuts_seps}, we need two auxiliary lemmas and a corollary.

\begin{lemma}\label{lm:no_loner}
Let $F$ be any minimal separator of $\LG(\Prof)$ and $u$ be any vertex of any input tree. Then, $\Inc(u) \not \subseteq F$.
\end{lemma}

\begin{proof}
Suppose $F$ is a minimal $a$-$b$ separator of $\LG(\Prof)$ and $u$ is a vertex of some input tree such that $\Inc(u) \subseteq F$. Consider any vertex $e \in \Inc(u)$. Then, there exists a path $\pi$ from $a$ to $b$ in $\LG(\Prof)$ where $e$ is the only vertex of $F$ in $\pi$. If such a path $\pi$ did not exist, then $F-e$ would still be a $a$-$b$ separator, and $F$ would not be minimal, a contradiction. Let $e_1$ and $e_2$ be the neighbors of $e$ in $\pi$ and let $e = \{u, v\}$. Since $\Inc(u) \subseteq F$, $\pi$ does not contain any other vertex $e'$ where $u \in e'$.  Thus, $e \cap e_1 = \{v\}$ and $e \cap e_2 = \{v\}$. Let $\pi = a, \dots, e_1, e, e_2, \dots, b$. Then $\pi' = a, \dots, e_1, e_2, \dots, b$ is also a path from $a$ to $b$. But $\pi'$ does not contain any vertex of $F$, contradicting the assumption that $F$ is a separator of $\LG(\Prof)$.
Hence, neither such a minimal separator $F$ nor such a vertex $u$ exist. \qed
\end{proof}

\begin{lemma}\label{lm:two_components}
If $F$ is a minimal separator of $\LG(\Prof)$, then $\LG(\Prof)-F$ has exactly two connected components.
\end{lemma}

\begin{proof}
Assume that $\LG(\Prof)-F$ has more than two connected components. By Lemma~\ref{lm:two_full}, $\LG(\Prof)-F$ has at least two full components. Let $H_1$ and $H_2$ be two full components of $\LG(\Prof)-F$. Let $H_3$ be a connected component of $\LG(\Prof)-F$ different from $H_1$ and $H_2$. By assumption $\LG(\Prof)$ is connected. Thus, there exists an edge $\{e, e_3\}$ in $\LG(\Prof)$ where $e \in F$ and $e_3 \in H_3$. Since $H_1$ and $H_2$ are full components, there exist edges $\{e, e_1\}$ and $\{e, e_2\}$ in $\LG(\Prof)$ where $e_1 \in V(H_1)$ and $e_2 \in V(H_2)$. 

Let $e = \{u, v\}$, and assume without loss of generality that $u \in e \cap e_3$. Then, there is no vertex $f \in V(H_1)$ where $u \in e \cap f$. Thus, $v \in e \cap e_1$. Similarly, 
there is no vertex $f \in V(H_2)$ such that $u \in f \cap e$ or $v \in f \cap e$. But then $H_2$ does not contain a vertex adjacent to $e$, so $H_2$ is not a full component, a contradiction.
\qed
\end{proof}

\begin{corollary}\label{cor:sep_connected}
If $F$ is a minimal separator of $\LG(\Prof)$, then $\LG(\Prof)-F'$ is connected for any $F' \subset F$.
\end{corollary}

\paragraph{Proof of Lemma~\ref{lm:cuts_seps}.}
We prove that if $F$ is a legal minimal separator of $\LG(\Prof)$ then $F$ is a nice minimal cut of $G(\Prof)$. The proof for the other direction is similar and is omitted. 

First, we show that $F$ is a cut of $G(\Prof)$. Assume the contrary. Let $\{u, v\}$ and $\{p, q\}$ be vertices in different components of $\LG(\Prof)-F$. Since $G(\Prof)-F$ is connected, there exists a path between vertices $u$ and $q$. Also, $\{u,v\} \notin F$ and $\{p,q\} \notin F$. Thus, by Observation~\ref{obs:path} there also exists a path between vertices $\{u, v\}$ and $\{p, q\}$ of $\LG(\Prof)-F$. This implies that $\{u,v\}$, $\{p,q\}$ are in the same connected component of $\LG(\Prof)-F$, a contradiction. Thus $F$ is a cut.
 
Next we show that $F$ is a nice cut of $G(\Prof)$. For every $T \in \Prof$ all the vertices of $\LG(T)$ in $F$ form a clique in $\LG(T)$. Thus, all the edges of $T$ in $F$ are incident on a common vertex, so $F$ is a legal cut. To complete the proof, assume that $G(\Prof)-F$ has a connected component with no edge and let $u$ be the vertex in one such component. Then, $\Inc(u) \subseteq F$. But $F$ is a minimal separator of $\LG(\Prof)$, and by Lemma~\ref{lm:no_loner}, $\Inc(u) \not \subseteq F$, a contradiction. Thus, $F$ is a nice cut.

Lastly, we show that $F$ is a minimal cut of $G(\Prof)$. Assume, on the contrary, that there exists $F' \subset F$ where $G(\Prof)-F'$ is disconnected. Since $F' \subset F$ and every connected component of $G(\Prof)-F$ has at least one edge, every connected component  of $G(\Prof)-F'$ also has at least one edge. Let $\{u,v\}$ and $\{p, q\}$ be the edges in different components of $G(\Prof)-F'$. By Corollary~\ref{cor:sep_connected}, $\LG(\Prof)-F'$ is connected and thus, there is a path between $\{u,v\}$ and $\{p, q\}$ in $\LG(\Prof)-F'$.  By Observation~\ref{obs:path} there must also be a path between vertices $u$ and $p$ in $G(\Prof)-F'$. Hence, edges $\{u, v\}$ and $\{p, q\}$ are in the same connected component of $G-F'$, a contradiction. Thus, $F$ is a minimal cut.
\qed

\begin{lemma}\label{lm:parallel}
Two legal minimal separators $F$ and $F'$ of $\LG(\Prof)$ are parallel if and only if the nice minimal cuts $F$ and $F'$ are parallel in $G(\Prof)$.
\end{lemma}
\begin{proof}
Assume that legal minimal separators $F$ and $F'$ of $\LG(\Prof)$ are parallel, but nice minimal cuts $F$ and $F'$ of $G(\Prof)$ are not. Then, there exists $\{\{u, v\}, \{p, q\}\} \subseteq F'$ where $\{u, v\}$ and $\{p, q\}$ are in different components of $G(\Prof)-F$. Since $F$ and $F'$ are parallel separators in $\LG(\Prof)$, and $F$ does not contain $\{u, v\}$ and $\{p, q\}$, there exists a path between vertices $\{u, v\}$ and $\{p, q\}$ in $\LG(\Prof)-F$. Then, by Observation~\ref{obs:path} there also exists a path between vertices $u$ and $q$ in $G(\Prof)-F$. Thus, edges $\{u, v\}$ and $\{p, q\}$ are in the same connected component of $G(\Prof)-F$, a contradiction.

The other direction can be proved similarly, using Observation~\ref{obs:path}. \qed
\end{proof}

The next lemma, from~\cite{Gysel2012}, follows from the definition of restricted triangulation.

\begin{lemma}\label{lm:two_cliques}
Let $H$ be a restricted triangulation of $\LG(\Prof)$ and let $(T, B)$ be a clique tree of $H$. Let $e = \{u, v\}$ be any vertex in $\LG(\Prof)$.  Then, there does not exist a node $x \in V(T)$ where $B(x)$ contains vertices from both $\Inc(u)\setminus e$ and $\Inc(v)\setminus e$.
\end{lemma}

\begin{lemma}\label{lm:two_subtrees}
Let $T$ be a tree in $\Prof$ and suppose $F$ is a minimal cut of $G(\Prof)$ that contains precisely one edge $e$ of $T$. Then, the edges of the two subtrees of $T-e$ are in different connected components of $G(\Prof)-F$.
\end{lemma}
\begin{proof}
Since $F$ is a minimal cut of $G(\Prof)$, the endpoints of $e$ are in different connected components of $G(\Prof)-F$. 
Let $e=\{u,v\}$. For every $x \in e$, let $T_x$ represent the subtree containing vertex $x$ in $T-e$. Edge $e$ is the only edge of $T$ in $F$. Thus, for every $x \in e$ all the edges of $T_x$ are in the same connected component of $G(\Prof)-F$ as vertex $x$. Since the endpoints of $e$ are in different connected components of $G(\Prof)-F$, the edges of $T_u$ and $T_v$ are also in different connected components of $G(\Prof)-F$.
\qed
\end{proof}

\section{Characterizing Compatibility via Cuts}

\label{sec:comp_cuts}

A set $\F$ of cuts of $G(\Prof)$ is \emph{complete} if, for every input tree $T \in \Prof$ and every internal edge $e$ of $T$, there exists a cut $F \in \F$ where $e$ is the only edge of $T$ in $F$.

\begin{lemma}\label{lm:complete}
$G(\Prof)$ has a complete set of pairwise parallel nice minimal cuts if and only if it has a complete set of pairwise parallel legal minimal cuts.
\end{lemma}
\begin{proof}
The ``only if part'' follows from the definition of a nice cut. Let $\F$ be a complete set of pairwise parallel legal minimal cuts. Consider any minimal subset $\F'$ of $\F$ that is also complete. Let $F$ be a legal minimal cut of $\F'$. Since $\F'$ is minimal, there exists an edge $e \in F$ of some input tree $T$ such that $e$ is the only edge of $T$ in $F$. Also, since $e$ is an internal edge, both the subtrees of $T-e$ have at least one edge each. Thus by Lemma~\ref{lm:two_subtrees}, both the connected components of $G(\Prof)-F$ have at least one edge each. Hence, $F$ is a nice minimal cut of $G(\Prof)$. It thus follows that $\F'$ is a complete set of pairwise parallel nice minimal cuts of $G(\Prof)$.
\qed
\end{proof}

We now characterize the compatibility of a profile in terms of minimal cuts in the display graph of the profile.

\begin{theorem}\label{thm:cuts_equal_compatibility}
A profile $\Prof$ of unrooted phylogenetic trees is compatible if and only if there exists a complete set of pairwise parallel legal minimal cuts for $G(\Prof)$.
\end{theorem}

\begin{example}\label{ex:compatible_cuts}
For the display graph of Fig.~\ref{fig:example}, let $\F = \{F_1,F_2,F_3,F_4\}$, where $F_1 = \{\{1,2\}, \{5,6\}\}$, $F_2 = \{\{2,3\}, \{6,7 \},\{5,6\}\}$, $F_3=\{\{4,5\}, \{1,2\},\{1,c\}\}$ and $F_4=\{\{6,7\},\{2,f\}\}$. Then, $\F$ is a complete set of pairwise parallel nice minimal cuts.
\end{example}

Theorem~\ref{thm:cuts_equal_compatibility}  and Lemmas~\ref{lm:cuts_seps},~\ref{lm:parallel}, and~\ref{lm:complete} imply an analogous result for $\LG(\Prof)$. A set $\F$ of legal minimal separators of $\LG(\Prof)$ is \emph{complete}, if for every internal edge $e$ of an input tree $T$, there exists a separator $F \in \F$ where $e$ is the only vertex of $\LG(T)$ in $F$. 

\begin{theorem}\label{thm:seps_equal_compatibility}
A profile $\Prof$ of unrooted phylogenetic trees is compatible if and only if there exists a complete set of pairwise parallel legal minimal separators for $\LG(\Prof)$.
\end{theorem}


Theorem~\ref{thm:cuts_equal_compatibility} follows from Theorem~\ref{thm:edge_label}, Lemma~\ref{lm:complete}, and the next result.

\begin{lemma}\label{lm:maximal_seps_maximal_cuts}
The following two statements are equivalent.
\begin{enumerate}[(i)]
\item
There exists a maximal set $\F$ of pairwise parallel minimal separators of $\LG(\Prof)$ where every separator in $\F$ is legal.
\item
There exists a complete set of pairwise parallel nice minimal cuts for $G(\Prof)$.
\end{enumerate}
\end{lemma}

\begin{proof}
\textit{(i) $\Rightarrow$ (ii):}
We show that for every internal edge $e=\{u,v\}$ of an input tree $T$ there exists a minimal separator in $\F$ that contains only vertex $e$ from $\LG(T)$. Then it follows from Lemmas~\ref{lm:cuts_seps} and~\ref{lm:parallel} that $\F$ is a complete set of pairwise parallel nice minimal cuts for display graph $G(\Prof)$.

As shown in the proof of Theorem~\ref{thm:edge_label}, $\LG(\Prof)_{\F}$ is a restricted minimal triangulation of $\LG(\Prof)$. Let $(S, B)$ be a clique tree of $\LG(\Prof)_{\F}$. By definition, the vertices in each of the sets $\Inc(u)$ and $\Inc(v)$ form a clique in $\LG(\Prof)$. Consider any vertex $p$ of $S$ where $\Inc(u) \subseteq B(p)$ and any vertex $q$ of $S$ where $\Inc(v) \subseteq B(q)$. (Since $(S, B)$ is a clique tree of $\LG(\Prof)_{\F}$, such vertices $p$ and $q$ must exist.) Also, by Lemma~\ref{lm:two_cliques}, $p \neq q$, $B(p) \cap (\Inc(v) \setminus \{e\}) = \emptyset$ and $B(q) \cap (\Inc(u) \setminus \{e\}) = \emptyset$.

Let $\pi = p, x_1, x_2, \dots, x_m, q$ be the path from $p$ to $q$ in $S$ where $m \geq 0$. Let $x_0=p$ and $x_{m+1}=q$. Let $x_i$ be the vertex nearest to $p$ in path $\pi$ where $i \in [m+1]$ and $B(x_i) \cap (\Inc(u) \setminus \{e\}) = \emptyset$. Let $F = B(x_{i-1}) \cap B(x_i)$. Then by Theorem~\ref{thm:parallel_minimal_ct}, $F \in \F$. Since $\Inc(u) \cap \Inc(v) = \{e\}$, by the coherence property, $e \in B(x_j)$ for every $j \in [m]$. Thus, $e \in F$. By Lemma~\ref{lm:two_cliques}, $B(x_{i-1}) \cap (\Inc(v) \setminus \{e\}) = \emptyset$. Since $B(x_i) \cap (\Inc(u) \setminus \{e\}) = \emptyset$, $F \cap \Inc(u) = \{e\}$ and $F \cap \Inc(v) = \{e\}$. Thus, for every vertex $e' \in \LG(T)$ where $e \neq e'$ and $e \cap e' \neq \emptyset$, $e' \notin F$. Also, since every separator in $\F$ is legal, we have $f \notin F$ for every vertex $f \in \LG(T)$ where $f \cap e = \emptyset$. Thus, $e$ is the only vertex of $\LG(T)$ in $F$. 

\vspace{5pt}



\noindent\textit{(i) $\Leftarrow$ (ii):}
Consider any complete set of pairwise parallel nice minimal cuts $\F'$ of $G(\Prof)$. By Lemmas~\ref{lm:cuts_seps} and~\ref{lm:parallel}, $\F'$ is a set of pairwise parallel legal minimal separators of $\LG(\Prof)$. There exists a maximal set $\F$ of pairwise parallel minimal separators where $\F' \subseteq \F$. 

Assume that $\F \setminus \F'$ contains a minimal separator $F$ that is not legal.  Then, there must exist a tree $T \in \Prof$ where at least two nonincident edges $e_1=\{x, y\}$ and $e_2=\{x', y'\}$ of $T$ are in $F$. Consider any internal edge $e_3$ in $T$ where $e_1$ and $e_2$ are in different components of $T-e_3$. Such an edge exists because $e_1$ and $e_2$ are nonincident. Since $\F'$ is complete, there exists a cut $F' \in \F'$ where $e_3$ is the only edge of $T$ in $F'$. Since $F$ and $F'$ are in $\F$, they are parallel to each other and vertices $e_1$ and $e_2$ are in the same connected component of $\LG(\Prof)-F'$. Thus, by Observation~\ref{obs:path}, there exists a path between vertices $x$ and $x'$ in $G(\Prof)-F'$ and edges $e_1$ and $e_2$ are also in the same connected component of $G(\Prof)-F'$. But by Lemma~\ref{lm:two_subtrees} that is impossible. 

Thus, every separator of $\F \setminus \F'$ is legal and $\F$ is a maximal set of pairwise minimal separators of $\LG(\Prof)$ where every separator in $\F$ is legal.
\qed
\end{proof}

\section{Splits and Cuts}

\label{sec:splits_cuts}

We first argue that for every nice minimal cut of $G(\Prof)$ we can derive a split of $\Leaves(\Prof)$.

\begin{lemma}\label{lm:cut_gives_split}
Let $F$ be a nice minimal cut of $G(\Prof)$ and let $G_1$ and $G_2$ be the two connected components of $G(\Prof)-F$. Then, $\Leaves(G_1)| \Leaves(G_2)$ is a split of $\Leaves(\Prof)$. 
\end{lemma}

\begin{proof}
Consider $G_i$ for each $i \in \{1,2\}$. We show that $\Leaves(G_i)$ is non-empty. Since $F$ is nice, $G_i$ contains at least one edge $e$ of $G(\Prof)$. If $e$ is a non-internal edge, then $\Leaves(G_i)$ is non-empty.  Assume that $e=\{u,v\}$ is an internal edge of some input tree $T$. If $F$ does not contain an edge of $T$, then $\Leaves(T) \subseteq \Leaves(G_i)$ and thus $\Leaves(G_i)$ is non-empty. Assume that $F$ contains one or more edges of $T$. Let $T_u$, $T_v$ be the two subtrees of $T-e$. Since $F$ is a nice minimal cut, $F$ contains edges from either $T_u$ or $T_v$ but not both. Without loss of generality assume that $F$ does not contain edges from $T_u$. Then, every edge of $T_u$ is in the same component as $e$. Since $T_u$ contains at least one leaf, $\Leaves(G_i)$ is non-empty.  
Thus, $\Leaves(G_1)|\Leaves(G_2)$ is a split of $\Leaves(\Prof)$.
\qed
\end{proof}

Let $\sigma(F)$ denote the split of $\Leaves(\Prof)$ induced by a  nice minimal cut $F$.  If $\F$ is a set of nice minimal cuts of $G(\Prof)$, $\Sigma(\F)$ denotes the set of all the non-trivial splits in $\bigcup_{F \in \F}\sigma(F)$.  The following result expresses the relationship between complete sets of nice minimal cuts and the compatibility of splits.

\begin{theorem}\label{lm:splits_compatibility}
If $G(\Prof)$ has a complete set of pairwise parallel nice minimal cuts  $\F$, then $\Sigma(\F)$ is compatible and any  compatible tree for $\Sigma(\F)$ is also a compatible tree for $\Prof$.
\end{theorem}

\begin{example}
For the cuts of the display graph in Fig.~\ref{fig:example} given in Example 1, we have $\sigma(F_1) =abc | defg$, $\sigma(F_2) = abcfg | de$, $\sigma(F_3) = ab | cdefg$, and $\sigma(F_4) = abcde | fg$.  Note that these splits are pairwise compatible.
\end{example}

The proof of Theorem~\ref{lm:splits_compatibility} uses the following lemma.

\begin{lemma}\label{lm:parallel_implies_compatible}
Let $F_1$ and $F_2$ be two parallel nice minimal cuts of $G(\Prof)$. Then, $\sigma(F_1)$ and $\sigma(F_2)$ are compatible.
\end{lemma}
\begin{proof}
Let $\sigma(F_1) = U_1| U_2$ and $\sigma(F_2) = V_1|V_2$. Assume that $\sigma(F_1)$ and $\sigma(F_2)$ are incompatible. Thus, $U_i \cap V_j \neq \emptyset$ for every $i, j \in \{1,2\}$. Let $a \in U_1 \cap V_1$, $b \in U_1 \cap V_2$, $c \in U_2 \cap V_1$ and $d \in U_2 \cap V_2$.  Since $\{a, b\} \subseteq U_1$, there exists a path $\pi_1$ between leaves $a$ and $b$ in $G(\Prof)-F_1$. But $a$ and $b$ are in different components of $G(\Prof)-F_2$. Thus, an edge $e_1$ of path $\pi_1$ is in the cut $F_2$. Similarly, $\{c,d\} \subseteq U_2$ and there exists a path $\pi_2$ between labels $c$ and $d$ in $G(\Prof)-F_1$. Since $c$ and $d$ are in different components of $G(\Prof)-F_2$, cut $F_2$ contains an edge $e_2$ of path $\pi_2$. But $\pi_1$ and $\pi_2$ are in different components of $G(\Prof)-F_1$, so edges $e_1$ and $e_2$ are in different components of $G(\Prof)-F_1$. Since $\{e_1, e_2\} \subseteq F_2$, the cuts $F_1$ and $F_2$ are not parallel, a contradiction.
\qed
\end{proof}

\paragraph{Proof of Theorem~\ref{lm:splits_compatibility}.}
The compatibility of  $\Sigma(\F)$ follows from Lemma~\ref{lm:parallel_implies_compatible} and Theorem~\ref{thm:SET}.  
Let $S$ be a compatible tree for $\Sigma(F)$, let $T$ be an input tree of $\Prof$, let $S' = S_{|\Leaves(T)}$, and let  $e$ be any internal edge of $T$. We now show that $S'$ displays $\sigma(e)$. 

Let $\sigma(e) = A|B$. There exists a cut $F \in \F$ where $e$ is the only edge of $T$ in $F$.  By Lemma~\ref{lm:two_subtrees}, since $F$ is minimal, the leaves of sets $A$ and $B$ are in different components of $G(\Prof)-F$. Thus, if $\sigma(F) = A'|B'$ then up to renaming of sets we have $A \subseteq A'$ and $B \subseteq B'$. Because $S$ displays $\sigma(F)$, $S'$ also displays $\sigma(e)$. Since $S'$ displays all the splits of $T$, $T$ can be obtained from $S'$ by contracting zero or more edges~\cite{SempleSteel03}.  Thus, $S$ displays $T$. Since $S$ displays every tree in $\Prof$, $S$ is a compatible tree for $\Prof$.
\qed

\section{Characterizing Agreement via Cuts}

\label{sec:agree_cuts}

The following characterization of agreement is similar to the one for tree compatibility given by Theorem~\ref{thm:cuts_equal_compatibility}, except for an additional restriction on the minimal cuts. 

\begin{theorem} \label{thm:AST_cuts}
A profile $\Prof$ has an agreement supertree if and only if $G(\Prof)$ has a complete set $\F$ of pairwise parallel legal minimal cuts where, for every cut $F \in \F$ and for every $T \in \Prof$, there is at most one edge of $T$ in $F$.
\end{theorem}

\begin{example}
One can verify that the display graph of Fig.~\ref{fig:example} does not meet the conditions of Theorem \ref{thm:AST_cuts} and, thus, the associated profile does not have an AST. On the other hand, for the display graph of Fig.~\ref{fig:example1}, let $\F = \{F_1, F_2, F_3\}$, where $F_1 = \{\{1,2\}, \{4, 5\}\}$, $F_2 = \{\{1,2\}, \{5, 6\}\}$ and $F_3 = \{\{2, 3\}, \{6, d\}\}$. For any given input tree $T$, every cut in $\F$ has at most one edge of $T$. Also, $\F$ is a complete set of pairwise parallel legal minimal cuts. Thus, by Theorem~\ref{thm:AST_cuts}, the input trees of Fig.~\ref{fig:example1} have an AST.
\end{example}

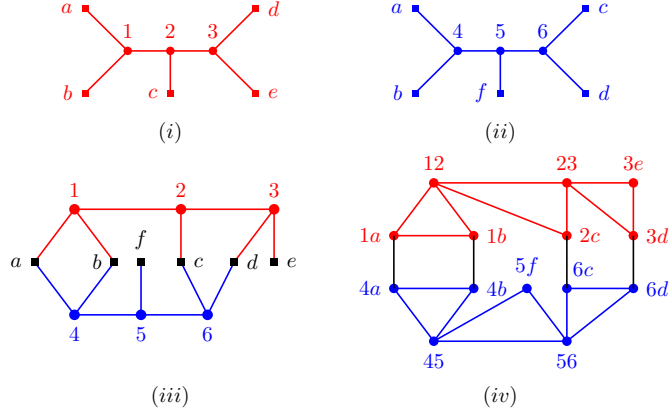
\begin{figure}[t!]
\renewcommand{\figurename}{Fig.}
 \begin{center}
\begin{tikzpicture} [inner sep = 2mm, semithick,scale=0.7]
\tikzstyle{every node}=[font=\fontsize{10}{10}, scale=0.8];

\coordinate (one_a) at (1,8);
\begin{scope}[red,scale=0.8]
\draw (one_a) --++(1, -1) coordinate(one_1) --++(-1,-1) coordinate(one_b)  ++(1,1) --++(1,0) coordinate(one_2) --++(0,-1) coordinate(one_c) ++(0,1) --++(1,0) coordinate (one_3) --++(1,1) coordinate(one_d) ++(-1,-1) --+(1,-1) coordinate(one_e);
\path [fill] (one_1) circle(0.1) node[above] {$1$};
\path [fill] (one_2) circle(0.1) node[above] {$2$};
\path [fill] (one_3) circle(0.1) node[above] {$3$};
\path [fill] (one_a) node[left] {$a$} ++(-0.075,-0.075)  rectangle +(0.15,0.15);
\path [fill] (one_b) node[left] {$b$} ++(-0.075,-0.075)  rectangle +(0.15,0.15);
\path [fill] (one_c) node[left] {$c$} ++(-0.075,-0.075)  rectangle +(0.15,0.15);
\path [fill] (one_d) node[right]{$d$}++(-0.075,-0.075)  rectangle +(0.15,0.15);
\path [fill] (one_e) node[right]{$e$}++(-0.075,-0.075)  rectangle +(0.15,0.15);
\end{scope}

\path(one_c) ++(0, -0.75) coordinate(i)node{$(i)$};
\path (one_d) ++(3,0) coordinate(two_a);
\begin{scope}[blue,scale=0.8]
\draw (two_a) --++(1,-1) coordinate(two_4) --++(-1,-1) coordinate(two_b) ++(1,1) --++(1,0) coordinate(two_5)--++(0,-1) coordinate(two_f) ++(0,1) --++(1,0) coordinate(two_6) --++(1,1) coordinate(two_c) ++(-1,-1) --+(1,-1) coordinate(two_d);
\begin{scope}
\path [fill] (two_4) circle(0.1) node[above] {$4$};
\path [fill] (two_5) circle(0.1) node[above] {$5$};
\path [fill] (two_6) circle(0.1) node[above] {$6$};
\end{scope}
\path [fill] (two_a) node[left] {$a$} ++(-0.075,-0.075)  rectangle +(0.15,0.15);
\path [fill] (two_b) node[left] {$b$} ++(-0.075,-0.075)  rectangle +(0.15,0.15);
\path [fill] (two_c) node[right] {$c$} ++(-0.075,-0.075)  rectangle +(0.15,0.15);
\path [fill] (two_d) node[right] {$d$} ++(-0.075,-0.075)  rectangle +(0.15,0.15);
\path [fill] (two_f) node[left] {$f$} ++(-0.075,-0.075)  rectangle +(0.15,0.15);
\end{scope}
\path (two_f) ++(0,-0.75) coordinate(ii)node{$(ii)$};
\path (one_1) ++(-1,-3) coordinate(dis_1);

\begin{scope}[red]
\draw (dis_1) --++(-0.75, -1) coordinate(dis_a) ++(0.75,1) -- ++(0.75,-1) coordinate(dis_b) ++(-0.75,1) -- ++(2,0) coordinate(dis_2) -- ++(0, -1) coordinate(dis_c) ++(0,1)--++(1.75,0) coordinate(dis_3) -- ++(-0.75,-1) coordinate(dis_d) ++(0.75,1) -- ++(0, -1) coordinate (dis_e);
\begin{scope}
\path [fill] (dis_1) circle(0.1) node[above] {$1$};
\path [fill] (dis_2) circle(0.1) node[above] {$2$};
\path [fill] (dis_3) circle(0.1) node[above] {$3$};
\end{scope}

\end{scope}

\begin{scope}[blue]
\draw (dis_a) --++(0.75,-1) coordinate(dis_4)--++(0.75,1) ++(-0.75,-1) -- ++(1.25,0) coordinate(dis_5) -- ++(0,1) coordinate(dis_f) ++(0,-1)--++(1.25,0) coordinate(dis_6)--(dis_c);
\draw (dis_6)--(dis_d);
\begin{scope}
\path [fill] (dis_4) circle(0.1) node[below] {$4$};
\path [fill] (dis_5) circle(0.1) node[below] {$5$};
\path [fill] (dis_6) circle(0.1) node[below] {$6$};
\end{scope}

\end{scope}

\begin{scope}[black]
\path [fill] (dis_a) node[left] {$a$} ++(-0.075,-0.075)  rectangle +(0.15,0.15);
\path [fill] (dis_b) node[left] {$b$} ++(-0.075,-0.075)  rectangle +(0.15,0.15);
\path [fill] (dis_c) node[right] {$c$} ++(-0.075,-0.075)  rectangle +(0.15,0.15);
\path [fill] (dis_d) node[right] {$d$} ++(-0.075,-0.075)  rectangle +(0.15,0.15);
\path [fill] (dis_e) node[right] {$e$} ++(-0.075,-0.075)  rectangle +(0.15,0.15);
\path [fill] (dis_f) node[above] {$f$} ++(-0.075,-0.075)  rectangle +(0.15,0.15);
\end{scope}

\path (i) ++(0, -5) node{$(iii)$};
\path (dis_3) ++(3,0.5) coordinate(el_12);

\begin{scope}[red]
\draw (el_12)--++(-0.75, -1) coordinate(el_1a) -- ++(1.5,0) coordinate(el_1b)--++(-0.75,1)--++(2.5,0) coordinate(el_23)--++(0,-1) coordinate(el_2c) ++(0,1) --++(1.25,-1) coordinate(el_3d)--++(0,1) coordinate(el_3e)--+(-1.25,0);
\draw (el_2c)--(el_12);

\begin{scope}
\path [fill] (el_12) circle(0.09) node[above]{$12$};
\path [fill] (el_1a) circle(0.09) node[left]{$1a$};
\path [fill] (el_1b) circle(0.09) node[right]{$1b$};
\path [fill] (el_23) circle(0.09) node[above]{$23$};
\path [fill] (el_2c) circle(0.09) node[right]{$2c$};
\path [fill] (el_3e) circle(0.09) node[above]{$3e$};
\path [fill] (el_3d) circle(0.09) node[right]{$3d$};
\end{scope}

\end{scope}

\path (el_1a) ++(0,-1) coordinate(el_4a);
\begin{scope}[blue]
\draw (el_4a) --++(1.5,0)coordinate(el_4b)--++(-0.75,-1) coordinate(el_45) --(el_4a);
\draw (el_45)--++(2.5,0) coordinate(el_56)--++(-0.75,1) coordinate(el_5f)-- (el_45); 
\draw (el_56)--++(0,1) coordinate(el_6c)--++(1.25,0) coordinate(el_6d)--(el_56);

\begin{scope}
\path [fill] (el_4a) circle(0.09) node[left]{$4a$};
\path [fill] (el_4b) circle(0.09) node[right]{$4b$};
\path [fill] (el_45) circle(0.09) node[below]{$45$};
\path [fill] (el_5f) circle(0.09) node[above]{$5f$};
\path [fill] (el_56) circle(0.09) node[below]{$56$};
\path [fill] (el_6c) circle(0.09) ++(0.3,0) node[above]{$6c$};
\path [fill] (el_6d) circle(0.09) node[right]{$6d$};
\end{scope}
\end{scope}

\begin{scope}[black]
\draw (el_1a)--(el_4a);
\draw (el_1b)--(el_4b);
\draw (el_2c)--(el_6c);
\draw (el_3d)--(el_6d);
\end{scope}

\path (ii) ++(0,-5) node{$(iv)$};

\end{tikzpicture}
 \end{center}
\vspace{-0.3cm}
\caption{\emph{(i)} First input tree. \emph{(ii)} Second input tree, which agrees with the first. \emph{(iii)} Display graph of the input trees. \emph{(iv)} Edge label intersection graph of the input trees, where label $uv$ represents edge $\{u,v\}$ of the display graph.}
 \label{fig:example1}
\end{figure}


The analogue of Theorem~\ref{thm:AST_cuts} for $\LG(\Prof)$ stated next follows from Theorem~\ref{thm:AST_cuts} and Lemmas~\ref{lm:cuts_seps},~\ref{lm:parallel}, and~\ref{lm:complete} .

\begin{theorem}\label{thm:AST_seps}
A profile $\Prof$ has an agreement supertree if and only if $\LG(\Prof)$ has a complete set $\F$ of pairwise parallel legal minimal separators where, for every $F \in \F$ and every $T \in \Prof$, there is at most one vertex of  $\LG(T)$ in $F$.
\end{theorem}

Theorem~\ref{thm:AST_cuts} follows from Lemma~\ref{lm:complete} and the next result.

\begin{lemma}\label{lm:AST_cuts}
A profile $\Prof$ has an agreement supertree if and only if $G(\Prof)$ has a complete set $\F$ of pairwise parallel nice minimal cuts where, for every cut $F \in \F$ and every $T \in \Prof$, there is at most one edge of $T$ in $F$.
\end{lemma}

The rest of the section is devoted to the proof of Lemma~\ref{lm:AST_cuts}. 

Let $S$ be an AST of $\Prof$ and let $e=\{u,v\}$ be an edge of $S$.  Let $S_u$ and $S_v$ be the subtrees of $S-e$  containing $u$ and $v$, respectively.  Let $L_u = \Leaves(S_u)$ and $L_v = \Leaves(S_v)$. Thus, $\sigma_e(S)= L_u|L_v$. Assume that there exists an input tree $T$ where $\Leaves(T) \cap L_x \neq \emptyset$ for each $x \in \{u,v\}$. Then there exists an edge $f \in E(T)$ where, if $\sigma_{f}(T)=A_1|A_2$, then $A_1 \subseteq L_u$ and $A_2 \subseteq L_v$. (If there were no such edge, $S_{|\Leaves(T)}$ would contain a split that is not in $T$ and would thus not be isomorphic to $T$.) We call $e$ an \emph{agreement edge} of $S$ corresponding to edge $f$ of $T$. Note that there does not exist any other edge $f'$ of $T$ where $e$ is also an agreement edge of $S$ with respect to edge $f'$ of $T$.

Given an AST $S$ of $\Prof$, we define a function $\Psi$ from $E(S)$ to subsets of edges of $G(\Prof)$ as follows. For every $e \in E(S)$, an edge $f$ of an input tree $T$ is in $\Psi(e)$  if and only if $e$ is an agreement edge of $S$ corresponding to edge $f$ of $T$. Observe that $\Psi$ is uniquely defined. We call $\Psi$ the \emph{cut function} of $S$. Given an edge $e \in E(S)$, we define a set $V_x$ for every $x \in e$ as follows. For every $T \in \Prof$, $V_x$ contains all the vertices of the minimal subtree of $T$ connecting the labels in $\Leaves(T) \cap L_x$. Note that if $e=\{u,v\}$ then $\{V_u, V_v\}$ is a partition of $V(G(\Prof))$.

\begin{lemma}\label{lm:psi_is_cut}
Let $S$ be an AST of $\Prof$ and let $\Psi$ be the cut function of $S$. Then,
\begin{enumerate}[(i)]
\item for every edge $e \in E(S)$, $\Psi(e)$ is a cut of $G(\Prof)$ and
\item for any edge $e \in E(S)$, $\Psi(e)$ is a minimal cut of $G(\Prof)$ if and only if $G(\Prof)-\Psi(e)$ has exactly two connected components.
\end{enumerate}
\end{lemma}
\begin{proof}
$(i)$ Let $e=\{u,v\}$. We show that $G(\Prof)-\Psi(e)$ does not contain an edge whose endpoints are in distinct sets of $\{V_u, V_v\}$. Assume the contrary. Let $f=\{x,y\}$ be an edge of $G(\Prof)-\Psi(e)$ where $x \in V_u$ and $y \in V_v$. 
Since $f \in G(\Prof) - \Psi(e)$, $f \notin \Psi(e)$. Suppose $f$ is an edge of input tree $T$. There are two cases.
\begin{enumerate}
\item \emph{$\Psi(e)$ does not contain an edge of $T$.} Then, there exists an endpoint $p$ of $e$ where $\Leaves(T) \subseteq L_p$. Without loss of generality, let $u = p$. Then, $V(T) \subseteq V_u$ and thus $y \in V_u$, a contradiction. 
\item \emph{$\Psi(e)$ contains an edge $f' \neq f$ of $T$.} Let $f'=\{r,s\}$ and let $L_r \subseteq L_u$ and $L_s \subseteq L_v$. Let $x$,$r$ be the vertices of $f$ and $f'$ where $L_x \subset L_r$. Since $T$ is a phylogenetic tree, such vertices $x$ and $r$ exist. Since $L_r \subseteq L_u$, both the endpoints of $f$ are in $V_u$, a contradiction.
\end{enumerate}
Thus, $G(\Prof)-\Psi(e)$ does not contain an edge whose endpoints are in different sets of $\{V_u, V_v\}$. Since $V_u$ and $V_v$ are non-empty, it follows that $\Psi(e)$ is a cut of $G(\Prof)$.

$(ii)$ The ``only if'' part follows from the definition of a minimal cut. We now prove the ``if'' part. Let $e=\{u,v\}$. Assume that $G(\Prof)-\Psi(e)$ has exactly two connected components. From the proof of $(i)$, $V_u$ and $V_v$ are the vertex sets of those two connected components. Consider any edge $f \in \Psi(e)$. The endpoints of $f$ are in different sets of  $\{V_u, V_v\}$ and thus are in different connected components of $G(\Prof)-\Psi(e)$. This implies that $G(\Prof)-(\Psi(e) \setminus \{f\})$ is connected. Thus, if $G(\Prof)-\Psi(e)$ has exactly two connected components, $\Psi(e)$ is a minimal cut of $G(\Prof)$.
\qed 

\hspace{0cm}
\end{proof}

Let $S$ be an AST of $\Prof$ and let $e$ be an edge of $S$. Although the preceding result shows that $\Psi(e)$ is a cut of $G(\Prof)$, $\Psi(e)$ may not be minimal.  We now argue that we can always construct an agreement supertree whose cut function gives minimal cuts. 

Suppose $e = (u,v)$ is a an edge of $S$ where $\Psi(e)$ is not minimal.   Let $\{L_1, \dots, L_m\}$ be the partition of $L_v$ where for every $i \in [m]$, $L_i =\Leaves(C) \cap L_v$ for some connected component $C$ in $G(\Prof)-\Psi(e)$.  We assume without loss of generality that $m > 1$ (if not, we can just exchange the roles of $u$ and $v$).  
Let $R_v$ be the rooted tree derived from $S_v$ by distinguishing vertex $v$ as the root.  Let  $R_{v,i}$ be the (rooted) tree obtained from the minimal subtree of $R_v$ connecting the labels in $L_i$ by distinguishing the vertex closest to $v$ as the root and suppressing every other vertex that has degree two. 
To \emph{split edge $e$ at $u$} is to construct a new tree $S'$ from $S$ in two steps:  (i) delete the vertices of $R_v$ from $S$ and (ii) for every $i \in [m]$, add an edge from $u$ to the root of $R_{v,i}$.

We can show the following by repeatedly splitting edges that do not correspond to minimal cuts.  For brevity, we omit the proof.

\begin{lemma}\label{lm:minimal_tree}
If $\Prof$ has an AST, then it has an AST $S$ of $\Prof$ whose cut function $\Psi$ satisfies the following: For every edge $e \in S$, $\Psi(e)$ is a minimal cut of $G(\Prof)$.
\end{lemma}

\paragraph{Proof of Lemma~\ref{lm:AST_cuts}.}
($\Leftarrow$) Assume that $\Prof$ has an AST.  Then, by Lemma~\ref{lm:minimal_tree},  $\Prof$ has an AST $S$ whose cut function $\Psi$ has the property that, for every edge $e \in E(S)$, $\Psi(e)$ is a minimal cut of $G(\Prof)$. Let $\F$ be the set of all $\Psi(e)$ such that $e$ is an internal edge of $S$. Then, $\F$ is a set of minimal cuts of $G(\Prof)$.  Further, by definition of $\Psi$, for every $F \in \F$ and for every $T \in \Prof$, $F$ contains at most one edge of $T$. Thus every cut in $\F$ is legal. We now prove that $\F$ is a complete set of pairwise parallel nice minimal cuts of $G(\Prof)$. 

We first prove that every cut in $\F$ is nice. Consider any $F \in \F$. Let $e=\{u,v\}$ be the internal edge of $S$ where $\Psi(e)=F$. Let $T$ be an input tree that has an internal edge $f$ in $\Psi(e)$. Since $e$ is an internal edge at least one such input tree exists; otherwise $\Psi(e)$ is not a minimal cut. Now, by definition, $f$ is the only edge of $T$ in $\Psi(e)$, so, by Lemma~\ref{lm:two_subtrees}, each of the two connected components of $G(\Prof)-\Psi(e)$ has at least one non-internal edge of $T$. Hence, $F$ is a nice minimal cut of $G(\Prof)$.

To prove that the cuts in $\F$ are pairwise parallel, we argue that for any two distinct internal edges $e_1$ and $e_2$ of $S$, $\Psi(e_1)$ and $\Psi(e_2)$ are parallel. There exist vertices $x \in e_1$ and $y \in e_2$ where $L_x \subseteq L_y$.  For every edge $f \in \Psi(e_1)$, we show that $f \in \Psi(e_2)$ or $f \subseteq V_y$.
It then follows that $\Psi(e_1)$ and $\Psi(e_2)$ are parallel. Let $f$ be an edge of input tree $T$.
Then there exists $z \in f$ where $L_z \subseteq L_x$. Thus, $L_z \subseteq L_y$ and $z \in V_y$. By Lemma~\ref{lm:psi_is_cut}, all the vertices of $V_y$ are in the same connected component of $G(\Prof)-\Psi(e_2)$. Thus, $f \in \Psi(e_2)$ or $f \subseteq V_y$. 

Lastly, we show that $\F$ is complete. Consider any internal edge $f=\{p,q\}$ of some input tree $T$. Since $S$ is an AST of $\Prof$, there exists an edge $e=\{u,v\}$ where, up to relabeling of sets, $L_p \subseteq L_u$ and $L_q \subseteq L_v$. Thus, $e$ is an agreement edge of $S$ corresponding to $f$, so $f \in \Psi(e)$. Since $f$ is an internal edge, $e$ is also an internal edge of $S$ and thus $\Psi(e) \in \F$. Hence, for every internal edge $f$ of an input tree there is a cut $F \in \F$ where $f \in F$. Thus, $\S$ is complete.


($\Rightarrow$) Assume that there exists a complete set $\F$ of pairwise parallel nice minimal cuts of $G(\Prof)$ where, for every $F \in \F$ and every $T  \in \Prof$, $F$ contains at most one edge of $T$. By Theorem~\ref{lm:splits_compatibility}, $\Sigma(\F)$ is compatible and, by Theorem~\ref{thm:SET}, there exists an unrooted tree $S$ where $\Sigma(\F)=\Sigma(S)$.  We prove that $S$ is an AST of $\Prof$ by showing that $\Sigma(S_{|\Leaves(T)}) = \Sigma(T)$ for every input tree $T \in \Prof$.


Consider an input tree $T$ of $\Prof$. Let $X_1|X_2$ be the non-trivial split of $T$ corresponding to edge $f \in E(T)$. Since $\F$ is complete, there exists a cut $F \in \F$ where $f \in \F$. If $\sigma(F)=Y_1|Y_2$, by Lemma~\ref{lm:two_subtrees}, up to relabeling of sets, $X_i \subseteq Y_i$ for every $i \in \{1,2\}$. Since $\sigma(F)$ is a split of $S$, this implies that $\Sigma(T) \subseteq \Sigma(S_{|\Leaves(T)})$.

Consider any non-trivial split $P_1|P_2$ of $\Sigma(S)$ where $P_i \cap \Leaves(T) \neq \emptyset$ for each $i \in \{1,2\}$. Let $Q_i = P_i \cap \Leaves(T)$ for each $i \in \{1,2\}$. Since $\Sigma(S)=\Sigma(\F)$, there exists a cut $F \in \F$ where $\sigma(F)=P_1|P_2$. 
Since $P_1$ and $P_2$ are in different connected components of $G(\Prof)-F$, $Q_1$ and $Q_2$ are also in different connected components of $G(\Prof)-F$. Thus, there exists an edge $f'$ of $T$ in $F$. Since $F$ does not contain any other edge of $T$, $\sigma(f')=Q_1|Q_2$. Thus, $\Sigma(S_{|\Leaves(T)}) \subseteq \Sigma(T)$.
\qed

\section{Conclusion}
We have shown that the characterization of tree compatibility in terms of restricted triangulations of the edge label intersection graph transforms into a characterization in terms of minimal cuts in the display graph. These two characterizations are closely related to the legal triangulation characterization of~\cite{Vakati11}. 
We also derived characterizations of the agreement supertree problem in terms of minimal cuts and minimal separators of the display and edge label intersection graphs respectively.

It is not known if the agreement supertree problem is fixed parameter tractable when parametrized by the number of input trees. It remains to be seen whether any of these characterizations can lead to explicit fixed parameter algorithms for the tree compatibility and agreement supertree problems when parametrized by the number of trees. 
\paragraph{Acknowledgment.}
We thank Sylvain Guillemot for his valuable comments. 

\bibliographystyle{abbrv}
\bibliography{cuts_modified}

\vfill

\pagebreak

\end{document}